\documentclass{article}
\usepackage{amsmath,amssymb,mathrsfs,amsthm}
\usepackage[utf8]{inputenc}
\usepackage{tikz}
\usepackage{stmaryrd}
\usepackage[all]{xy}

\newcommand{\f}{\mathbb{F}_p}
\newcommand{\ff}{\mathbb{F}_q}
\newcommand{\fff}{\mathbb{F}_{q^k}}
\newcommand{\cc}{\mathbb{C}}
\newcommand{\vek}[1]{\boldsymbol{\mathrm{#1}}}
\newcommand{\codec}{\mathscr{C}}
\newcommand{\coded}{\mathscr{D}}
\newcommand{\lt}{\mathcal{L}(\mathscr{C},\f)}
\newcommand{\s}{\mathcal{S}}
\newcommand{\mc}[1]{\mathcal{#1}}
\newcommand{\sbar}{\overline{\s}}
\newcommand{\perps}{\perp_s}
\newcommand{\seq}[3]{(\vek{#2}_1,\ldots,\vek{#2}_{#1}\mid\vek{#3}_1,\ldots,\vek{#3}_{#1})}
\newcommand{\ro}[1]{\rho_{\vek #1}}
\newcommand{\BH}{{\rm BH}}
\DeclareMathOperator{\tr}{tr}
\DeclareMathOperator{\spn}{span}
\DeclareMathOperator{\fix}{Fix}
\DeclareMathOperator{\swt}{swt}
\DeclareMathOperator{\wt}{wt}
\DeclareMathOperator{\stab}{Stab}

\newcommand{\bil}[1]{#1\times #1\longrightarrow R}

\DeclareMathOperator{\Hom}{Hom}

\DeclareMathOperator{\Tr}{Tr}

\DeclareMathOperator{\blf}{BLF}
\DeclareMathOperator{\Aut}{Aut}
\usetikzlibrary{quantikz}
\title{A Construction of Quantum Stabilizer Codes from Classical Codes and Butson Hadamard Matrices}
\author{Bülent Sara\c{c}, Damla Acar}
\date{}
\newtheorem{thm}{Theorem}[section]
\newtheorem{cor}[thm]{Corollary}
\newtheorem{prop}[thm]{Proposition}
\newtheorem{lem}[thm]{Lemma}

\newtheorem{example}{Example}

\theoremstyle{remark}
\newtheorem{rem}{Remark}
\numberwithin{equation}{section}
\begin{document}
\maketitle
\begin{abstract}
    In this paper, we give a constructive proof to show that if there exist a classical linear code $\codec\subseteq\ff^n$ of dimension $k$ and a classical linear code $\coded\subseteq\fff^m$ of dimension $s$, where $q$ is a power of a prime number $p$, then there exists an $\llbracket nm,ks,\delta\rrbracket_q$ quantum stabilizer code with $\delta$ determined by $\codec$ and $\coded$ by identifying the stabilizer group of the code. In the construction, we use a particular type of Butson Hadamard matrices equivalent to multiple Kronecker products of the Fourier matrix of order $p$. We also consider the same construction of a quantum code for a general normalized Butson Hadamard matrix and search for a condition for the quantum code to be a stabilizer code. 
\end{abstract}
\section{Introduction and Preliminaries}
Quantum error-correcting codes have experienced rapid growth since their inception by Shor in \cite{shor1995scheme}. One of the most important classes of these codes is quantum stabilizer codes, first introduced by Gottesman \cite{gottesman1997stabilizer} and Calderbank et al. \cite{calderbank1996good}. Stabilizer codes have been extensively studied due to their relatively simple encoding algorithms and their structural advantages, which allow for connections to classical error-correcting codes for analysis.

In this paper, we propose a method for constructing quantum error-correcting codes that generalize Shor's code and the quantum codes discussed in \cite{FLZ}. Our approach encompasses both binary and non-binary schemes. We use quantum digits (qudits for short) over $\ff$ as a unit of quantum information, where $\ff$ denotes the finite field with $q$ elements. The states of qudits are simply vectors in $\mathbb{C}^q$. This vector space is equipped with the standard inner product, with respect to which there is an orthonormal basis, whose elements are denoted $\ket x$, where $x$ runs through the elements of $\ff$. Now the state of a system of $n$ qudits are represented by vectors in ${(\cc^{q})}^{\otimes n}$. The set $\{\ket{x_1}\otimes\cdots \otimes \ket{x_n}:x_1,\ldots,x_n\in\ff\}$ forms an orthonormal basis for the Hilbert space  $(\cc^q)^{\otimes n}$. For the sake of simplicity, we denote the basis vector  $\ket{x_1}\otimes\cdots \otimes \ket{x_n}$ as $\ket{x_1\ldots x_n}$. Thus every vector in $(\cc^q)^{\otimes n}$ is a linear combination of the vectors $\ket{\vek{x}}$, where $\vek x\in\ff^n$.

A quantum code of length $n$ is a  nonzero subspace of ${(\cc^{q})}^{\otimes n}$.  We denote a $q$-ary quantum code which encodes $k$ qudits of information into $n$-qudits as   $\llbracket n,k\rrbracket_q$.

 In order to provide information about the types of errors a quantum code can detect or correct, we first define the unitary operators on $\mathbb{C}^q$ by 
\begin{equation*}
    X(a)\ket{x}=\ket{x+a},\; Z(b)\ket{x}=\omega^{\tr( bx)}\ket{x}
\end{equation*}
where $a,b,x\in\ff$, $\tr$ denotes the trace function from $\ff$ to its prime field $\mathbb{F}_p$, and $\omega$ is a primitive $p$-th root of unity. Note that the operators $X(a)$ and $Z(b)$, for nonzero $a$ and $b$, coincide with the usual bit flip and phase flip errors, respectively, on (binary) quantum bits (or qubits). Let $\mathcal{E}=\{X(a)Z(b)\; |\; a,b\in\ff\}$. Then $\mathcal{E}$ forms a basis for the vector space of linear operators on $\cc^q$. Further, we observe that the following properties are satisfied, which make $\mathcal{E}$ a nice error basis on $\cc^q$ (see \cite{knill1996non}):
\begin{enumerate}
      \item[(NE1)] $\mathcal{E}$ contains the identity operator.
    \item[(NE2)] $X(a)Z(b)X(a')Z(b')=\omega^{\tr(ba')}X(a+a')Z(b+b')$ for every $a,a',b,b'\in\ff$. Therefore the product of two elements of $\mathcal{E}$ is a scalar multiple of another element of $\mathcal{E}$.
    \item[(NE3)] If $M,N$ are distinct elements of $\mathcal{E}$ then $\Tr(M^{\dagger}N)=0$. 
\end{enumerate}
\begin{rem}
    Let $p=2$. If we take $\omega=-1$ for the binary case  it results in omission of  complex phases. Thus we use $\omega$ as a primitive $4$-th root of unity; more precisely, the imaginary unit $i$.  
\end{rem}

One can extend $\mathcal{E}$ to a suitable error basis on $n$ qudits, as follows.  For $\vek{a}=(a_1,\ldots,a_n)\in\ff^{n}$, we write $X(\vek{a})=X(a_1)\otimes\cdots\otimes X(a_n)$ and $Z(\vek a)=Z(a_1)\otimes\cdots\otimes Z(a_n)$ for the tensor products of $n$ error operators. Then  $\mathcal{E}_n=\{X(\vek{a})Z(\vek{b})\mid \vek{a},\vek{b}\in \ff^{n}\}$ forms a nice error basis on $(\mathbb{C}^{q})^{\otimes n}$ (see \cite{ketkar2006nonbinary}). 

Binary stabilizer codes were first introduced by  Gottesman in \cite{gottesman1997stabilizer}. They have proved to form an important class of quantum codes because they provide an algebraic approach for error correction. Later, the stabilizer formation was extended to non-binary quantum error-correcting codes in \cite{AK2001} and investigated in detail in \cite{KKKS}. 

Let $\mathcal P_n=\{\omega^cX(\vek a)Z(\vek b):\vek a,\vek b\in\fff,c\in\mathbb{F}_p\}$. Then $\mathcal{P}_n$ is a finite group of order $pq^{2n}$ for $p>2$. Note that when $p=2$, we take $\omega=i$ and $c\in\{0,1,2,3\}$; hence $|\mathcal{P}_n|=4q^{2n}$. $\mathcal{P}_n$ is called the error group associated with $\mathcal{E}_n$ (see \cite{KKKS}).  

Given a qunatum code $Q$ of length $n$, one may think of the subgroup $$\stab(Q)=\{E\in\mathcal{P}_n:E\vek v=\vek v\text{ for all }\vek v\in Q\}$$ of $\mathcal{P}_n$, called the \emph{stabilizer group} of $Q$. Note that $\stab(Q)$ is an abelian subgroup of $\mathcal{P}_n$ with $\stab(Q)\cap \mathcal Z(\mathcal{P}_n)=\{I\}$. Conversely, given an abelian subgroup $\mathcal{S}$ of $\mathcal{P}_n$ with $\mc{S}\cap \mc Z(\mathcal{P}_n)=\{I\}$, one can define the quantum code $$\fix(\mc S)=\{\vek v\in (\cc^q)^{\otimes n}:E\vek v=\vek v\text{ for all }E\in\mc S\},$$the joint eigenspace of $\mc S$ associated with the eigenvalue 1. Observe that if $Q$ is a quantum code, then $Q\subseteq \fix(\stab(Q))$. If the equality holds, then we call $Q$ a \emph{stabilizer code}. Equivalently, a quantum code $Q$ is a stabilizer code if and only if there exists an abelian subgroup $\mc S$ of $\mc P_n$ with $\mc S\cap \mc Z(\mc P_n)=\{I\}$ such that $Q=\fix(\mc S)$. If $Q$ is a $q$-ary stabilizer code of length $n$, then $\mc S=\stab(Q)$ is a finite abelian subgroup of $\mc P_n$ of order $q^n/\dim(Q)$. It follows that if $\mc S$ has a minimal generating set (of $r$ elements, say), then $|\mc S|=q^r$ since $\mc S$ is abelian and every non-central element of $\mc P_n$ has order $q$, which gives that $Q$ is an $\llbracket n,n-r\rrbracket_q$-code.

For a subgroup $\mc S$ of $\mc P_n$, we write $C_{\mc P_n}(\mc S)$ to denote the centralizer of $\mc S$ in $\mc P_n$, namely $C_{\mc P_n}(\mc S)=\{E\in\mc P_n:ES=SE\text{ for every }S\in\mc S\}$. Also, the subgroup of $\mc P_n$ generated by $\mc S$ and $\mc Z(\mc P_n)$ is $\mc S\mc Z(\mc P_n)$. The following proposition, proved in \cite{ketkar2006nonbinary}, shows that the error detection capability of a stabilizer code is closely related to the subgroups $C_{\mc P_n}(\mc S)$ and $\mc S\mc Z(\mc P_n)$.
\begin{prop}\label{cond_detectable_error}
    Suppose that $\mc S\leq \mathcal{P}_n$ is the stabilizer group of a stabilizer code $Q$ with $\dim Q>1$. An error $E$ in $\mathcal{P}_n$ is detectable by the quantum code $Q$ if and only if either $E$ is an element of $\mc S\mc Z(\mathcal{P}_n)$ or $E$ does not belong to the centralizer $C_{\mathcal{P}_n}(\mc S)$.
\end{prop}
The weight of an error operator $E=\omega^{c}E_1\otimes\cdots\otimes E_n\in\mathcal{P}_n$ is defined by $\wt(E)=\#\{i \mid E_i\neq I\}$,
that is the number of tensor factors that are non-identity. We say that a quantum code $Q$ has \emph{minimal distance} $d$ if $d$ is a positive integer such that every error in $\mc P_n$ of weight less than $d$ can be detected by $Q$ whereas there is an error in $\mc P_n$ of weight $d$ which cannot be detected by $Q$. Thus, it follows easily from Proposition \ref{cond_detectable_error} that the minimum distance of a stabilizer code $Q$ with the stabilizer group $\mc S$ is given by
    \begin{equation*}
        d(Q)=   \left\{
\begin{array}{ll}
     \min\{\wt(E)\mid E\in C_{\mathcal{P}_n}(\mc S)\setminus \mc S \},& \text{if}\;   S\subsetneq C_{\mathcal{P}_n}(\mc S)\\
     \min \{\wt(E)\mid E\in \mc S\setminus\{I\}\}, & \text{if}\; \mc S=C_{\mathcal{P}_n}(\mc S)
\end{array} 
\right. 
    \end{equation*}
It is known that a code of minimum distance $d$ can correct all errors of weight at most $\lfloor (d-1)/2\rfloor$.

With any error of type $\omega^cX(\vek a)Z(\vek b)$, we associate the vector $(\vek a\mid\vek b)$ in $\mathbb F_q^{2n}$, where $\vek a,\vek b\in \ff^n$, and define a map $\psi$ by
\begin{equation}\label{def:psi}
    \psi : \mathcal{P}_{n}\rightarrow \ff^{2n}\;,\; \omega^{c}X(\vek{a})Z(\vek{b})\longmapsto(\vek{a}\mid \vek{b}).
\end{equation}
Note that $\psi$ is a group homomorphism from the multiplicative group $\mc P_n$ onto the additive group $\ff^{2n}$, through which some multiplicative properties in $\mc P_n$ translates into additive properties in $\ff^{2n}$. For instance, if $E=\omega^cX(\vek{a})Z(\vek{b})$ and $E'=\omega^{c'}X(\vek{a}')Z(\vek{b}')$ are the elements of $\mc P_n$, then the commutator $[E,E']=E^{-1}E'^{-1}EE'$ is equal to $\omega^{\tr(\vek{b}.\vek{a}'-\vek{b}'.\vek{a})}I$, and hence the errors $E$ and $E'$ commute if and only if $\tr(\vek b.\vek 
 a'-\vek b'.\vek a)=0$. This motivates the definition of a function $\langle \cdot\ ,\,\cdot\rangle_s : \ff^{2n}\times\ff^{2n}\rightarrow \f$  by
    \begin{equation*}
        \langle (\vek{a}\mid\vek{b}),(\vek{a'}\mid\vek{b'})\rangle_s:=\tr(\vek{b}.\vek{a'}-\vek{b'}.\vek{a}),
    \end{equation*}
which is an inner product on $\ff^{2n}$ as an $\mathbb F_p$--space, called the \emph{trace symplectic inner product}. 
It follows that the errors $E=\omega^cX(\vek{a})Z(\vek{b})$ and $E'=\omega^{c'}X(\vek{a'})Z(\vek{b'})$ in $\mc P_n$ commute if and only if the vectors $(\vek{a}\mid\vek{b})$ and $(\vek{a'}\mid\vek{b'})$ in $\ff^{2n}$ are orthogonal according to the trace symplectic inner product. We sometimes express the orthonality of the vectors $(\vek a\mid \vek b)$ and $(\vek a'\mid \vek b')$ by writing $(\vek a\mid \vek b)\perp_s (\vek a'\mid\vek b')$. 

Let $\overline{\mc S}$ be the image of a subgroup $\mc S$ of $\mc P_n 
 $under $\psi$. Then $\overline{\mc S}$ is an $\mathbb F_p$-subspace of $\ff^{2n}$. The symplectic dual of $\overline{\mc S}$ (denoted $\overline{\mc S}^{\perp_s}$) consists of the elements of $\ff^{2n}$ that are orthogonal to every element in $\overline{\mc S}$. Note that $\psi$ maps $C_{\mc P_n}(\mc S)$ onto $\overline{\mc S}^{\perp_s}$. In view of Proposition \ref{cond_detectable_error}, we define the \emph{symplectic weight} of a vector $(\vek a\mid\vek b)$ in $\ff^{2n}$ (denoted $\swt(\vek a\mid\vek b)$) to be the weight of the error $X(\vek a)Z(\vek b)$ in $\mc P_n$. Equivalently, $\swt(\vek a\mid\vek b)$ is the number of the $i$'s such that $(a_i,b_i)\neq(0,0)$ for $\vek a=(a_1,\ldots,a_n)$ and $\vek b=(b_1,\ldots,b_n)$.

A Butson-Hadamard (BH) matrix $H$ of order $n$ is an $n\times n$ matrix whose entries are primitive roots of unity such that $HH^{\dagger}=nI_n$, where $H^{\dagger}$ is the complex conjugate transpose of $H$ and $I_n$ is the identity matrix of order $n$. Thus the rows of a Butson Hadamard matrix are pairwise orthogonal with respect to the Hermitian inner product. If the entries of $H$ are $k$-th roots of unity then we say that $H$ is a $\BH(n,k)$ matrix. Butson Hadamard matrices were introduced by Butson in \cite{butson1962generalized} and has been studied in various aspects. Two Butson Hadamard matrices are said to be equivalent if they are of the same type, say BH$(n,k)$, and one can be obtained from the other by permuting rows (or columns) or multiplying rows (or columns) by a $k$-th root of unity. It is readily checked that every BH matrix is equivalent to a BH matrix whose first row and first column consist of 1's, which is called \emph{normalized}. One of the simplest examples of a normalized BH matrix is the Fourier matrix $[\zeta_n^{(i-1)(j-1)}]_{i,j=1}^n\in\text{BH}(n,n)$, where $\zeta_n$ is a primitive $n$-th rooth of unity. Moreover, one can construct a BH matrix $[\zeta_n^{\vek x\vek y^T}]_{\vek x,\vek y\in\f^n}\in\text{BH}(p^n,p)$, which is, indeed, the $n$-fold Kronecker product of the Fourier matrix of order $p$.

The organization of our paper is as follows. In the second section, we give our general method for constructing a quantum code using $p$-ary and $q$-ary classical codes and a BH$(q^k,p)$ matrix, where $p$ is a prime and $q$ is a power of $p$. In the third section, we look for a BH matrix for which our construction gives a stabilizer quantum code. In particular, we show that if the BH matrix used in the construction has a particular form such that it is equivalent to a Kronecker product of the Fourier matrix of order $p$, then the resulting quantum code is a stabilizer code. In the next section, we search for a converse, and consider a certain $q$-ary classical code in the construction. In particular, we prove that if the resulting quantum code is a stabilizer code, then the BH matrix used in the construction must be equivalent to a Kronecker product of the Fourier matrix of order  $p$. In Section 5, we give an algorithm of error correction for quantum stabilizer codes obtained by our construction as stabilizer codes. The last section is given as appendix devoted to provide a proof of the fact that the BH matrices considered in Section 3 and Theorem \ref{converse} are equivalent to a Knocker product of Fourier matrices.

\section{The General Construction}
 Let $\mathscr{C}\subseteq\ff^n$ be a $q$-ary $[n,k,d_1]$ classical linear code, where $1\le k<r$. We shall use the codewords of $\mathscr{C}$ to create quantum states in $(\mathbb{C}^q)^{\otimes n}$, as follows: Let $\{f_\lambda:\lambda\in\fff\}$ be a set of functions from $\mathscr{C}$ into $\f$ and let 
\begin{equation}\label{1.1}
    \phi_\lambda:=\frac{1}{\sqrt{q^k}}\sum_{\vek{c}\in\codec} \omega^{f_\lambda(\vek{c})}\ket{\vek{c}}
\end{equation}
for all $\lambda\in\fff$, where $\omega=e^{2\pi i/p}$. Note that the $\phi_\lambda$'s form an orthonormal basis for a subspace of $(\mathbb{C}^q)^{\otimes n}$ if and only if the $q^k\times q^k$ matrix 
\begin{equation}\label{1.2}
H=[\omega^{f_\lambda(\vek{c})}]_{{\lambda\in\fff, \vek{c}\in\codec}},
\end{equation}
with rows indexed by the elements of $\fff$ and columns indexed by the elements of $\codec$, both written in a fixed order, forms a BH$(q^k,p)$ matrix. 

Let $\coded\subseteq\fff^m$ be a classical linear code of dimension $s$ and define
\begin{equation}\label{1.3}
    \Phi_\Lambda:=\phi_{\lambda_1}\otimes\ldots\otimes\phi_{\lambda_m}\in(\mathbb{C}^q)^{\otimes nm}
\end{equation}
for all $\Lambda=(\lambda_1,\ldots,\lambda_m)\in\coded$. We form a quantum code, denoted $Q_H(\codec,\coded)$, of length $nm$ to be the subspace of $(\mathbb{C}^q)^{\otimes nm}$ spanned by $\Phi_\Lambda$ for all $\Lambda\in\coded$, i.e.,
\begin{equation}\label{1.4}
    Q_H(\codec,\coded):=\spn\{\Phi_\Lambda:\Lambda\in\coded\}.
\end{equation}
Note that $Q_H(\codec,\coded)$ is a $q^{ks}$-dimensional subspace of $(\mathbb{C}^q)^{\otimes nm}$, namely, $Q_H(\codec,\coded)$ is an $\llbracket nm,ks,\delta\rrbracket_q$ quantum code where $\delta$ is the minimum distance of $Q_H(\codec,\coded)$.
On the other hand, there exists a one-to-one correspondance $\nu:\ff^{ks}\rightarrow\coded$, and so one can set the logical state $\ket{\vek{a}}_L$ as $\Phi_{\nu(\vek{a})}$ for each $\vek{a}\in\ff^{ks}$. Then $Q_H(\codec,\coded)$ is a quantum code of length $nm$ that encodes $ks$ logical $q$-states. 
\begin{example}
If $p=q=2$, $\codec=\coded=\{000,111\}\subseteq\mathbb{F}_2^3$,  and 
\begin{equation*}
    H=
    \begin{pmatrix}
    1&1\\
    1&-1
    \end{pmatrix}
\end{equation*} 
in \eqref{1.2}, then the  construction of the quantum code $Q_H(\codec,\coded)$ as in \eqref{1.4} coincides with the well-known Shor's $9$-qubit code. 
\end{example}
\begin{example}\label{ex:qutrit}
    Let $p=q=3$, $\codec=\coded=\{000,111,222\}$, and $$H=\begin{pmatrix}
    1&1&1\\
    1&\omega&\omega^2\\
    1&\omega^2&\omega
    \end{pmatrix}$$ in \eqref{1.2}, where $w=e^{2\pi i/3}$. Then the quantum code obtained as in \eqref{1.4} is the same as the nine-qutrit error correcting code considered in Sect. V of \cite{ternary}.
    
    More generally, if $\codec=\coded=\{(\lambda,\ldots,\lambda):\lambda\in\ff\}\subseteq\ff^m$, the one-dimensional linear code over $\fff$ spanned by $(1,1,\ldots,1)\in\ff^m$, and $H$ is any BH$(q,p)$ matrix, then the quantum code $Q_H(\codec,\coded)$ in \eqref{1.4} is among the $\llbracket m^2,1,m\rrbracket_q$ quantum error-correcting codes studied in \cite{KLZ}. In view of the next section, we see that if the matrix $H$ is chosen to be a normalized Butson-Hadamard matrix of Fourier type, then $Q_H(\codec,\coded)$ turns out to be a stabilizer code.
\end{example}
\begin{rem}\label{rem1}
Let $\coded$ be a $q^k$-ary linear code of length $m$.
If there exists a positive integer $i$ with $1\le i\le m$ such that for every codeword $\Lambda$ in $\coded$, the $i$-th coordinate of $\Lambda$ is zero, then we can project $\coded$ onto a $q^k$-ary linear code $\coded'$ of length $m-1$ by deleting the $i$-th coordinate of each codeword of $\coded$, where the minimum distance remains unaltered, and use $\coded'$ instead of $\coded$ in the construction of $Q_H(\codec,\coded)$. Thus, throughout this note, we shall assume that the code $\coded$ in the above construction, is non-trivial (i,e., neither $0$ nor $\fff^m$) and that for every integer $1\le i \le m$, there exists a codeword in $\coded$ whose $i$-th coordinate is equal to $1$. 
\end{rem}
\begin{prop}
Let $\codec$ be a $q$-ary linear code of dimension $k$ and length $n$, and let $\coded$ be a $q^k$-ary linear code of dimension $s$ and length $m$ (where $0<s<m$). Let $H$ and $H'$ be BH$(q^k,p)$ matrices. If $Q_{H}(\codec,\coded)=Q_{H'}(\codec,\coded)$, then $H$ and $H'$ are row-equivalent BH matrices. The converse also holds when $\coded=\{(\lambda,\ldots,\lambda):\lambda\in\fff\}\subseteq\fff^m$ is the one-dimensional linear code over $\fff$.
\end{prop}
\begin{proof} Assume that $Q_{H}(\codec,\coded)=Q_{H'}(\codec,\coded)$.
Given $\lambda\in\fff$, we use the notations $\phi_\lambda$ or $\phi'_\lambda$ accordingly the coefficients of the quantum states in \eqref{1.1} come from $H$ or $H'$. Similarly, we write $\Phi'_\Lambda$ to mean the tensor products of states of the form $\phi'_\lambda$ in \eqref{1.3}.

Since the sets $\{\phi_\lambda:\lambda\in\fff\}$ and $\{\phi'_\lambda:\lambda\in\fff\}$ are orthonormal, they are independent; hence we have $$\spn\{\phi_\lambda:\lambda\in\fff\}=\spn\{\ket{\vek{c}}:\vek{c}\in\codec\}=\spn\{\phi'_\lambda:\lambda\in\fff\}.$$It follows that there exist $b_{\lambda\mu}\in\mathbb{C}$ ($\lambda,\mu\in\fff$) such that 
\begin{equation}\label{2.5}
\phi_\lambda=\sum_{\mu\in\fff}b_{\lambda\mu}\phi'_\mu 
\end{equation}
for all $\lambda\in\fff$. On the other hand, since $Q_{H}(\codec,\coded)=Q_{H'}(\codec,\coded)$, there exist $a_{\Lambda M}\in\mathbb{C}$ ($\Lambda,M\in\coded$) such that $\Phi_\Lambda=\sum_{M\in\coded}a_{\Lambda M}\Phi'_M$ for all $\Lambda\in\coded$. Rewriting $\Phi_\Lambda$ for $\Lambda=(\lambda_1,\ldots,\lambda_m)\in\coded$ as
\begin{align*}
\Phi_\Lambda&=\left(\sum_{\mu\in\fff}b_{\lambda_1\mu}\phi'_\mu\right)\otimes\cdots\otimes\left(\sum_{\mu\in\fff}b_{\lambda_1\mu}\phi'_\mu\right)\\
&=\sum_{(\mu_1,\ldots,\mu_m)\in\fff^m}\left(\prod_{i=1}^m b_{\lambda_i\mu_i}\right)(\phi'_{\mu_1}\otimes\cdots\otimes\phi'_{\mu_m}),
\end{align*}
we see that for all $\Lambda=(\lambda_1,\ldots,\lambda_m)\in\coded$, $\prod_{i=1}^m b_{\lambda_i\mu_i}=a_{\Lambda M}$ if $M=(\mu_1,\ldots,\mu_m)\in\coded$ and  $\prod_{i=1}^m b_{\lambda_i\mu_i}=0$ if $M=(\mu_1,\ldots,\mu_m)\not\in\coded$.

Considering the equation \eqref{2.5}, we see that there exists a function $\sigma:\fff\rightarrow\fff$ such that $b_{\lambda\sigma(\lambda)}\neq0$. Since we assume that $\coded$ is non-trivial, at least one of the standard basis element of $\fff^m$ does not belong to $\coded$. Without loss of generality, we assume that $(1,0,\ldots,0)\not\in\coded$. Let $\lambda\in\fff$. By our assumption on $\coded$ (see Remark \ref{rem1}), there exists $(\lambda_1,\ldots,\lambda_m)\in\coded$ with $\lambda_1=\lambda$. Then we must have $(\sigma(\lambda_1),\ldots,\sigma(\lambda_m))\in\coded$. Let $\alpha\in\fff\setminus \{0\}$. Then $(\alpha+\sigma(\lambda_1),\sigma(\lambda_2),\ldots,\sigma(\lambda_m))\not\in\coded$. This gives that $b_{\lambda\nu}.b_{\lambda_2\sigma(\lambda_2)}\ldots b_{\lambda_m\sigma(\lambda_m)}=0$, where $\nu=\alpha+\sigma(\lambda)$; hence $b_{\lambda\nu}=0$. It follows that $b_{\lambda\mu}=0$ for all $\mu\in\fff$ with $\mu\neq \sigma(\lambda)$. Therefore, $\phi_\lambda=b_{\lambda\sigma(\lambda)}\phi'_{\sigma(\lambda)}$ for all $\lambda\in\fff$, and so $\sigma$ is a permutation on $\fff$. It is now straightforward to check that $b_{\lambda\sigma(\lambda)}$ is a power of $\omega$ for all $\lambda\in\fff$. This completes the proof of the first assertion since we also have $H=[b_{\lambda\mu}]H'$. Now the second assertion follows since in  case  $\coded=\{(\lambda,\ldots,\lambda):\lambda\in\ff\}\subseteq\ff^m$, the row-equivalence of $H$ and $H'$ implies that   $\Phi_\Lambda$ is a constant multiple of $\Phi'_\Lambda$ for every $\Lambda\in\coded$.
\end{proof}
\section{A Quantum Stabilizer Code}
In this section, we use our general construction of quantum codes described in the preceding section to produce quantum stabilizer codes by choosing a particular set $\{f_\lambda:\lambda\in\fff\}$ of functions from $\codec$ into $\f$ in the formation of the $\phi_\lambda$'s  in \eqref{1.1}. We start with two lemmas.
\begin{lem}\label{lem:1}
Let $\codec\subseteq\ff^n$ be a linear code over $\ff$ and let $\vek{u}\in\ff^n$. Then $\tr_{q/p}(\vek{u}.\vek{c})=0$ for all $\vek{c}\in\codec$ if and only if $\vek{u}\in\codec^\perp$.
\end{lem}
\begin{proof}
It is enough to prove the ``only if" part of the statement. So, suppose that $\tr_{q/p}(\vek{u}.\vek{c})=0$ for all $\vek{c}\in\codec$. Choose an arbitrary nonzero element $\lambda\in\ff$. Then $\codec=\{\lambda\vek{c}:\vek{c}\in\codec\}$, and so $\tr_{q/p}(\lambda(\vek{u}.\vek{c}))=0$ for all $\vek{c}\in\codec$. Since $\lambda\in\ff$ is arbitrary, this implies that $\vek{u}.\vek{c}=0$ for all $\vek{c}\in\codec$; hence $\vek{u}\in\codec^\perp$. 
\end{proof}
\begin{lem}\label{lem:stab}
Let $\codec\subseteq\ff^n$ be a linear code over $\ff$ and let $\phi=\sum_{\vek{c}\in\codec}\omega^{f(\vek{c})}\ket{\vek{c}}\in(\mathbb{C}^q)^{\otimes n}$, where $f:\codec\rightarrow \f$ is a function. Let $\vek{u}\in\ff^n$. Then $Z(\vek{u})$ stabilizes $\phi$ if and only if $\vek{u}\in\codec^{\perp}$.
\end{lem}
\begin{proof}
This is clear by Lemma \ref{lem:1} since $Z(\vek{u})\phi=\sum_{\vek{c}\in\codec}\omega^{f(c)+\tr_{q/p}(\vek u.\vek c)}\ket{\vek c}.$
\end{proof}

Let $\mathscr{C}\subseteq\ff^n$ be a classical linear code over $\ff$ of dimension $k$, where $q=p^r$ and $1\le k<n$, and $\coded\subseteq\fff^m$ be a classical linear code over $\fff$ of dimension  $s$ with $1\le s<m$. Note that $\codec$ is also a vector space over $\f$ by restriction of scalars.  Let $\lt$ be the $\f$-dual of $\codec$. That is, $\lt$ is the set of all $\f$-linear transformations from $\codec$ to $\f$. Then $\lt$ is a vector space over $\f$ of dimension $rk$; in other words, $\lt$ contains $q^k$ linear transformations of $\f$-spaces. Also, there exists an $\f$-space isomorphism $\kappa:\fff\rightarrow\lt$. We set $f_\lambda=\kappa(\lambda)$ for each $\lambda\in\fff$, form the matrix $H=[\omega^{f_\lambda(\vek{c})}]_{\lambda\in\fff,\vek{c}\in\codec}$ which is necessarily a BH matrix, and define  $Q_H(\codec,\coded)$  as in \eqref{1.4}. As shown in Corollary \ref{last corollary}, $H$ is equivalent to the $rk$-fold Kronecker product of the Fourier matrix of order $p$.

Note that the above setting of $f_\lambda$'s yields that $f_{\lambda_1}+f_{\lambda_2}=f_{\lambda_1+\lambda_2}$ and $f_{c\lambda}=cf_{\lambda}$ for all $\lambda,\lambda_1,\lambda_2\in \fff$ and $c\in\f$. Given any $\Lambda=(\lambda_1,\ldots,\lambda_m)\in\coded$, define a mapping $F_\Lambda:\codec^m\rightarrow\f$ by $F_\Lambda(\vek{c}_1,\ldots,\vek{c}_m)=f_{\lambda_1}(\vek{c}_1)+\cdots+f_{\lambda_m}(\vek{c}_m)$ for all $(\vek{c}_1,\ldots,\vek{c}_m)\in\codec^m$. Then $F_\Lambda$ is a linear transformation of $\f$-spaces. Moreover, by definition of $f_\lambda$'s above, the set $\{F_\Lambda:\Lambda\in\coded\}$ is a  vector space over $\f$ in a natural way since $F_{\Lambda_1}+F_{\Lambda_2}=F_{\Lambda_1+\Lambda_2}$ and $h.F_\Lambda=F_{h\Lambda}$ for all $h\in\f$ and $\Lambda,\Lambda_1,\Lambda_2\in\coded$.

For $\vek x\in\ff^n$, define the mapping $\rho_{\vek x}:\codec\rightarrow\f$ by $\rho_{\vek x}(\vek c)=\tr_{q/p}(\vek c.\vek x)$. Then $\rho_{\vek x}$ is a linear transformation of $\f$-spaces,  and given $\vek x,\vek y\in\ff^n$, $\ro x=\ro y$ if and only if $\vek x-\vek y\in\codec^\perp$ by Lemma \ref{lem:1}. Thus we sometimes write $\rho_{\overline{\vek x}}$
for $\ro x$, where $\overline{\vek x}$ denotes the image of $\vek x$ under the canonical projection $\ff^n\rightarrow\ff^n/\codec^\perp$. Note that for each $\lambda\in\ff$, there corresponds $\vek x_\lambda$ such that $f_\lambda=\rho_{\overline{\vek x}_\lambda}$. This correspondence yields an $\f$-space isomorphism $\Theta:\fff\rightarrow\ff^n/\codec^\perp$, where $\Theta(\lambda)=\overline{\vek x}_\lambda$ for which $f_\lambda=\rho_{\overline{\vek x}_\lambda}$. Define $$\coded^\Theta:=\{(\vek{x}_{\lambda_1},\ldots,\vek{x}_{\lambda_m}):(\lambda_1,\ldots,\lambda_m)\in\coded\text{ and }\vek{x}_{\lambda_i}\in\Theta(\lambda_i)\text{ for all }1\le i\le m\}.$$ That is, $$\coded^\Theta=\bigcup_{(\lambda_1,\ldots,\lambda_m)\in\coded}\Theta(\lambda_1)\times\cdots\times\Theta(\lambda_m).$$ Then, clearly, $\coded^\Theta$ is an additive code over $\ff$. Note that all the vectors in $(\codec^\perp)^{(m)}$ are elements of $\coded^\Theta$ corresponding to the zero codeword in  $\coded$. Thus, $(\codec^\perp)^{(m)}\subseteq\coded^\Theta$.

Now we are ready to state and prove the main theorem of this section.
\begin{thm}\label{main_thm}
With the above notation, $Q_H(\codec,\coded)$ is an $\llbracket nm,ks,\delta\rrbracket_q$ quantum stabilizer code and $\delta=\min\{d(\codec),\ell\}$, where $\ell=\min\{\wt(\vek X):\vek X\in\coded^\Theta\setminus(\codec^\perp)^{(m)}\}$. Moreover, the stabilizer group $\mc S$ of $Q_H(\codec,\coded)$ consists of the errors of the form $X(\vek{c}_1,\ldots,\vek{c}_m)Z(\vek{d}_1,\ldots,\vek{d}_m)$, where $(\vek{c}_1,\ldots,\vek{c}_m)\in\bigcap_{\Lambda\in\coded}\ker(F_\Lambda)$ and $\vek{d}_1,\ldots,\vek{d}_m\in\codec^{\perp}$, and the centralizer $C_{\mc P_n}(\mc S)$ of $\mc S$  consists of the errors of the form $\omega^cX(\vek{u}_1,\ldots,\vek{u}_m)Z(\vek{v}_1,\ldots,\vek{v}_m)$, where $c\in\mathbb F_p$,  $\vek u_1,\ldots,\vek u_m\in\codec$ and $(\vek v_1,\ldots,\vek v_m)\in \coded^\Theta$.
\end{thm}
\begin{proof}

If $\vek{d}_1,\ldots,\vek{d}_m\in\codec^{\perp}$, then $Z(\vek{d}_1,\ldots,\vek{d}_m)\in\codec^\perp$ by Lemma \ref{lem:stab}. On the other hand, given $\vek e\in\codec$ and $\lambda\in\fff$, we have $$X(\vek e)\phi_\lambda=\sum_{\vek c\in\codec}\omega^{f_\lambda(\vek c)}\ket{\vek c+\vek e}=\sum_{\vek c\in\codec}\omega^{f_\lambda(\vek c-\vek e)}\ket{\vek c}=\omega^{-f_\lambda(e)}\phi_\lambda,$$and so $$X(\vek c_1,\ldots,\vek c_m)\Phi_\Lambda=\omega^{F_\Lambda(\vek c_1,\ldots,\vek c_m)}\Phi_\Lambda$$for all $\vek c_1,\ldots,\vek c_m\in\codec$ and $\Lambda\in\coded$. Thus, given codewords $\vek c_1,\ldots,\vek c_m$ of $\codec$, $X(\vek c_1,\ldots,\vek c_m)\in\s$ if and only if $(\vek c_1,\ldots,\vek c_m)\in\bigcap_{\Lambda\in\coded}\ker(F_\Lambda)$. Therefore, $\s$ contains all the errors of the form $X(\vek c_1,\ldots,\vek c_m)Z(\vek d_1,\ldots,\vek d_m)$ for which $\vek c_1,\ldots,\vek c_m\in\codec$ with $(\vek c_1,\ldots,\vek c_m)\in\bigcap_{\Lambda\in\coded}\ker(F_\Lambda)$ and $\vek d_1,\ldots,\vek d_m\in\codec^\perp$. Clearly, the number of the errors $Z(\vek d_1,\ldots,\vek d_m)$, where $\vek d_1,\ldots,\vek d_m\in\codec^\perp$ is equal to $|\codec^\perp|=q^{(n-k)m}$. We shall show that the number of errors $X(\vek c_1,\ldots,\vek c_m)$, where $(\vek c_1,\ldots,\vek c_m)\in\bigcap_{\Lambda\in\coded}\ker(F_\Lambda)$ is equal to $q^{k(m-s)}$.

Let $\{F_{\Lambda_1},\ldots,F_{\Lambda_{rks}}\}$ be an $\f$-basis for $\{F_\Lambda:\Lambda\in\coded\}$. Each $F_{\Lambda_i}$ can be represented by a $1\times rkm$ matrix, say $R_i$ with respect to a fixed ordered basis of $\codec^m$. Then the $rks\times rks$ matrix $$
\mathcal{R}=
\begin{pmatrix}
R_1\\ \vdots\\ R_{rks} 
\end{pmatrix}
$$
represents the $\f$-linear transformation $F:\codec^m\rightarrow\f^{rks}$ defined by $F(x)=(F_{\Lambda_1}(x),\ldots,F_{\Lambda_{rks}}(x))$ with respect to the same ordered basis of $\codec^m$. Since $$c_1R_1+\cdots+c_{rks}R_{rks}=0\text{ if and only if }c_1F_{\Lambda_1}+\cdots+c_{rks}F_{\Lambda_{rks}}=0$$for any $c_1,\ldots,c_{rks}\in\f$, we see that the set $\{R_1,\ldots,R_{rks}\}$ is linearly independent over $\f$. Thus $\mathcal{R}$ has rank $rks$, or equivalently, has nullity $rk(m-s)$. Since $\ker(F)=\bigcap_{i=1}^{rks}\ker(F_{\Lambda_i})$, $\dim_{\f}\left(\bigcap_{\Lambda\in\coded}\ker(F_{\Lambda})\right)=\dim_{\f}\left(\bigcap_{i=1}^{rks}\ker(F_{\Lambda_i})\right)=rk(m-s)$. Then the number of errors $X(\vek c_1,\ldots,\vek c_m)$, where $(\vek c_1,\ldots,\vek c_m)\in\bigcap_{\Lambda\in\coded}\ker(F_\Lambda)$ is equal to $p^{rk(m-s)}=q^{k(m-s)}$.

It follows that $|\s|\ge q^{(n-k)m}q^{k(m-s)}=q^{nm-ks}$. On the other hand, since $Q_H(\codec,\coded)\subseteq\fix(\s)$, we have $$q^{ks}=\dim_{\ff}(Q_H(\codec,\coded))\le\dim_{\ff}(\fix(\s))=\frac{q^{nm}}{|\s|},$$ and so $|\s|\le q^{nm-ks}$. This gives that $|\s|=q^{nm-ks}$, and hence $\s$ consists of the errors of the form $X(\vek c_1,\ldots,\vek c_m)Z(\vek d_1,\ldots,\vek d_m)$ for which $\vek c_1,\ldots,\vek c_m\in\codec$ with $(\vek c_1,\ldots,\vek c_m)\in\bigcap_{\Lambda\in\coded}\ker(F_\Lambda)$ and $\vek d_1,\ldots,\vek d_m\in\codec^\perp$. This shows, in particular, that $Q_H(\codec,\coded)$ is a stabilizer code.

Finally, we shall show that $\delta=\min\{d(\codec),\ell\}$. To see this, we first need to determine $\sbar^{\perps}$. We claim that $\sbar^{\perps}$ consists of $\seq{m}{u}{v}$ for which $\vek u_1,\ldots,\vek u_m\in\codec$ and $(\vek v_1,\ldots,\vek v_m)\in \coded^\Theta$. By above, we see that $\sbar$ consists of the sequences $\seq{m}{c}{d}$ such that $(\vek c_1,\ldots,\vek c_m)\in \bigcap_{\Lambda\in\coded}\ker(F_\Lambda)$ and $\vek d_1,\ldots,\vek d_m\in\codec^\perp$. Let $(\vek U\mid\vek V)=\seq{m}{u}{v}$ be such that  $\vek u_1,\ldots,\vek u_m\in\codec$ and $(\vek v_1,\ldots,\vek v_m)\in \coded^\Theta$. Let $(\vek C\mid\vek D)=\seq{m}{c}{d}\in\sbar$. By the choice of $\vek{v}_1,\ldots,\vek v_m$, there exists $\Lambda=(\lambda_1,\ldots,\lambda_m)\in\coded$ such that $\vek v_i\in\Theta(\lambda_i)$ for each $1\le i\le m$. In other words, $f_{\lambda_i}=\rho_{\vek v_i}$ for each $1\le i \le m$. This gives that 
\begin{equation}\label{eq:main_thm}
\tr_{q/p}(\vek C.\vek V)=\sum_{i=1}^m\tr_{q/p}(\vek c_i.\vek v_i)=\sum_{i=1}^m f_{\lambda_i}(\vek c_i)=F_\Lambda(\vek c_1,\ldots,\vek c_m)=0.
\end{equation}
It follows that $\langle(\vek C\mid\vek D),(\vek U\mid\vek V)\rangle_s=\tr_{q/p}(\vek d.\vek u-\vek v.\vek c)=0;$ hence $(\vek U\mid\vek V)\in\sbar^{\perp_s}$. Now let $(\vek U\mid\vek V)=\seq{m}{u}{v}\in\sbar^\perp$. Note that given any $\vek d\in\codec^\perp$, $(0,\ldots,0\mid \vek d, 0,\ldots,0),(0,\ldots,0\mid 0,\vek d, 0,\ldots,0),\ldots,(0,\ldots,0\mid 0,\ldots,0,\vek d)$ all lie in $\sbar$. Thus we have $\tr_{q/p}(\vek u_i.\vek d)=0$ for all $\vek d\in\codec^\perp$ and $1\le i\le m$. It follows, from Lemma \ref{lem:1}, that $\vek u_i\in\codec$ for all $1\le i\le m$. Now we shall show that $\vek v\in\coded^\Theta$. To see this, it is enough to show that $\coded^\Theta$ is equal to the additive code
\begin{equation}\label{eq:main_thm2}
    \mathcal{V}:=\{(\vek z_1,\ldots,\vek z_m):\sum_{i=1}^m \tr_{q/p}(\vek z_i.\vek c_i)=0\text{ for all }(\vek c_1,\ldots,\vek c_m)\in \bigcap_{\Lambda\in\coded}\ker(F_\Lambda)\}
\end{equation}
over $\ff$. By \eqref{eq:main_thm}, we see that $\coded^\Theta\subseteq\mathcal{V}$. To see the reverse inclusion, define $R_{\vek y_1,\ldots,\vek y_m}:\ff^{nm}\rightarrow\f$ by $R_{\vek y_1,\ldots,\vek y_m}(\vek z_1,\ldots,\vek z_m)=\tr_{q/p}\left(\sum_{i=1}^m{\vek y_i}.\vek z_i\right)$ for all $\vek y_1,\ldots,\vek y_m,\vek z_1,\ldots,\vek z_m\in\ff^n$. Note that $\mathcal{V}=\bigcap\{\ker(R_{\vek c_1,\ldots,\vek c_m}):(\vek c_1,\ldots,\vek c_m)\in\bigcap_{\Lambda\in\coded}\ker(F_\Lambda)\}$. Moreover, $\{R_{\vek c_1,\ldots,\vek c_m}:(\vek c_1,\ldots,\vek c_m)\in\bigcap_{\Lambda\in\coded}\ker(F_\Lambda)\}$ is an $\f$-space, in a natural way, and the correspondence $(\vek c_1,\ldots,\vek c_m)\mapsto R_{\vek c_1,\ldots,\vek c_m}$ is an $\f$-space isomorphism. Thus, by similar arguments as used above, one can see that the $\f$-dimension of $\mathcal{V}$ is equal to $$\dim_{\f}(\ff^{nm})-\dim_{\f}\left(\bigcap_{\Lambda	\in\coded}\ker(F_\Lambda)\right)=rnm-rkm+rks,$$ and so $|\mathcal{V}|=q^{(n-k)m}q^{ks}$. One can also see  that $|\coded^\Theta|=q^{(n-k)m}.q^{ks}$. Since we already have $\coded^\Theta\subseteq\mathcal{V}$, we must have the equality $\coded^\Theta=\mathcal{V}$. Therefore, $\sbar^{\perps}$ consists of $\seq{m}{u}{v}$ for which $\vek u_1,\ldots,\vek u_m\in\codec$ and $(\vek v_1,\ldots,\vek v_m)\in \coded^\Theta$. 

Let $\vek c\in\codec$ and suppose that $(\vek c,\vek 0,\ldots,\vek 0)\in\ker(F_\Lambda)$ for all $\Lambda\in\coded$. By our assumption on $\coded$ (see Remark \ref{rem1}) the first coordinates of the elements of $\coded$ form up $\fff$. It follows that $c\in\bigcap_{\lambda\in\fff}\ker(f_\lambda)=0$. This gives that for every nonzero $\vek c\in\codec$, the element $(\vek c,\vek 0,\ldots,\vek 0\mid\vek 0,\ldots,\vek 0)$ of $(\ff^n)^{2m}$ lies in $\sbar^{\perp_s}\setminus\sbar$. Therefore, $\swt(\sbar^{\perp_s}\setminus\sbar)\le d(\codec)$. On the other hand, for any nonzero $\Lambda=(\lambda_1,\ldots,\lambda_m)\in\coded$, an element $\vek D=(\vek{x}_{\lambda_1},\ldots,\vek{x}_{\lambda_m})\in\coded^\Theta$ cannot belong to $(\codec^\perp)^{(m)}$; hence the element $(\vek 0,\ldots,\vek 0\mid \vek D)$ of $(\ff^n)^{2m}$ belongs to $\sbar^{\perp_s}\setminus\sbar$. Therefore, we also have $\swt(\sbar^{\perp_s}\setminus\sbar)\le \ell$. Consequently, $\swt(\sbar^{\perp_s}\setminus\sbar)\le\min\{d(\codec),\ell\}$.

Now suppose that $\swt(\sbar^{\perp_s}\setminus\sbar)<\min\{d(\codec),\ell\}$. Let $(\vek C\mid\vek D)\in\sbar^{\perp_s}\setminus\sbar$ such that $\swt(\vek C\mid\vek D)=\swt(\sbar^{\perp_s}\setminus\sbar)$. Since $\swt(\vek C\mid\vek D)< d(\codec)$ and $\vek C\in\codec^{(m)}$, we must have $\vek C=\vek 0$. But since $\vek D\not\in(\codec^\perp)^{(m)}$, we get $\ell\le\swt(\vek C\mid\vek D)<\ell$, a contradiction. Therefore, $\swt(\sbar^{\perp_s}\setminus\sbar)=\min\{d(\codec),\ell\}$.
\end{proof}
\begin{cor}
    Let the situation be as in Theorem \ref{main_thm}. If $d(\codec)\le d(\coded)$, then $\delta=d(\codec)$.
\end{cor}
\begin{proof}
    It is not difficult to see that the number $\ell$ in Theorem \ref{main_thm} is at least $d(\coded)$. Now the result follows since $\delta=\min\{d(\coded),\ell\}$.
\end{proof}
\begin{cor}
Let the situation be as in Theorem \ref{main_thm}. Suppose that $\coded=\{(\lambda,\ldots,\lambda):\lambda\in\fff\}\subset\fff^m$. Then $\delta=\min\{d(\codec),m\}$ and the stabilizer group of $Q_H(\codec,\coded)$ consists of the errors $X(\vek{c}_1,\ldots\vek{c}_m)Z(\vek{d}_1\ldots\vek{d}_m)$, where $\vek c_1,\ldots,\vek c_m\in\codec$ with $\sum_{i=1}^m \vek c_i=0$ and $\vek d_1,\ldots,\vek d_m\in\codec^\perp$. 
\end{cor}
\begin{proof}
It is known from the proof of Theorem \ref{main_thm} that $\coded^\Theta$ is equal to the code  
$$\mathcal{V}:=\{(\vek z_1,\ldots,\vek z_m):\sum_{i=1}^m \tr_{q/p}(\vek z_i.\vek c_i)=0\text{ for all }(\vek c_1,\ldots,\vek c_m)\in \bigcap_{\Lambda\in\coded}\ker(F_\Lambda)\}$$ over $\ff$. Since $\coded=\{(\lambda,\ldots,\lambda):\lambda\in\fff\}\subset\fff^m$, $\bigcap_{\Lambda\in\coded}\ker(F_\Lambda)=\{(c_{1},\ldots,c_{m})\in\codec^{(m)} : 
 c_{1}+\cdots+c_{m}\in\bigcap_{\lambda\in\fff} f_{\lambda}=0\}$ such that $f_{\lambda}$'s are all linear transformations from $\codec$ to $\f$. So $\coded^\Theta=\{(c_{1},\cdots,c_{m})\in \codec^{(m)} : c_{1}+\cdots+c_{m}=0\}$. Therefore the stabilizer group of $Q_H(\codec,\coded)$ consists of the errors $X(\vek{c}_1,\ldots\vek{c}_m)Z(\vek{d}_1\ldots\vek{d}_m)$, where $\vek c_1,\ldots,\vek c_m\in\codec$ with $\sum_{i=1}^m \vek c_i=0$ and $\vek d_1,\ldots,\vek d_m\in\codec^\perp$.  For a suitable $\lambda$, $x_{\lambda}$  could be $(10\ldots 0)$. Then $\wt(\vek X)=m$, where  $\vek X$ is the $m$-fold tensor product of $x_{\lambda}$. The result follows since $\ell=m$.  
\end{proof}
\begin{prop}
    Let the situation be as in Theorem \ref{main_thm}, where $\coded=\{(\lambda,\ldots,\lambda):\lambda\in\fff\}\subset\fff^m$. Then
    $$Q_H(\codec,\coded)=\text{span}\left\{\sum_{\substack{(\vek c_1,\ldots,\vek c_m)\in\codec^{(m)}\\ \vek c_1+\cdots+\vek c_m=c}}\ket{\vek c_1\ldots\vek c_m}:\vek c\in\codec \right\}.$$
\end{prop}
\begin{proof}
    Let $$\mathcal{A}=\left\{\sum_{\substack{(\vek c_1,\ldots,\vek c_m)\in\codec^{(m)}\\ \vek c_1+\cdots+\vek c_m=c}}\ket{\vek c_1\ldots\vek c_m}:\vek c\in\codec \right\}.$$Since for any $\Lambda=(\lambda,\ldots,\lambda)\in\coded$, 
    \begin{align*}
        \Phi_\Lambda=\phi_\lambda^{\otimes m}&=\frac{1}{\sqrt{q^{km}}}\sum_{(\vek c_1,\ldots,\vek c_m)\in\codec^{(m)}}\omega^{f_\lambda(\sum_{i=1}^m \vek c_i)}\ket{\vek c_1\ldots\vek c_m}\\
        &=\frac{1}{\sqrt{q^{km}}}\sum_{(\vek c_1,\ldots,\vek c_m)\in\codec^{(m)}}\omega^{f_\lambda(\vek c)}\sum_{\vek c_1+\cdots+\vek c_m=\vek c}\ket{\vek c_1\ldots\vek c_m},
    \end{align*}
    we have $Q_H(\codec,\coded)\subseteq\spn \mathcal{A}$. Since $\dim(Q_H(\codec,\coded))=q^k=\dim(\spn \mathcal{A})$, we have the equality $Q_H(\codec,\coded)=\spn \mathcal{A}$.
\end{proof}
\section{In Search of a Converse}
\begin{lem}\label{lem:equality_of_tensors}
Let $\codec$ be a non-empty subset of $\ff^n$ and let $\phi_i=\sum_{\vek{c}\in\codec}\omega^{\alpha_i(\vek{c})}\ket{\vek{c}}$, $\psi_i=\sum_{\vek{c}\in\codec}\omega^{\beta_i(\vek{c})}\ket{\vek{c}}$ be elements of $(\mathbb{C}^q)^{\otimes n}$ for every $1\le i\le m$, where $\omega=e^{2\pi i/p}$ and the $\alpha_i$ and $\beta_i$ are functions from $\codec$ into $\f$. Suppose that $$\phi_1\otimes\cdots\otimes\phi_m=\psi_1\otimes\cdots\otimes\psi_m.$$Then there exist $b_r\in\f$ ($1\le r\le m$) such that $\psi_r=\omega^{b_r}\phi_r$ for all $1\le r\le m$.
\end{lem}
\begin{proof}
By assumption, we have the equality $$\sum_{(\vek{c}_1,\ldots,\vek{c}_m)\in\codec^{m}}\omega^{\sum_{i=1}^m\alpha_i(\vek{c}_i)}\ket{\vek{c}_1\ldots\vek{c}_m}=\sum_{(\vek{c}_1,\ldots,\vek{c}_m)\in\codec^{m}}\omega^{\sum_{i=1}^m\beta_i(\vek{c}_i)}\ket{\vek{c}_1\ldots\vek{c}_m},$$where $\codec^{m}$ denotes the $m$-fold Cartesian product of $\codec$. Hence $\sum_{i=1}^m\alpha_i(\vek{c}_i)=\sum_{i=1}^m\beta_i(\vek{c}_i)$ for every $(\vek{c}_1,\ldots,\vek{c}_m)\in\codec^{m}$. Fix a $\vek{c}_0\in\codec$. Then $$\alpha_r(\vek{c})+\sum_{\substack{i=1\\ i\neq r}}^m\alpha_i(\vek{c}_0)=\beta_r(\vek{c})+\sum_{\substack{i=1\\ i\neq r}}^m\beta_i(\vek{c}_0)$$for all $\vek{c}\in\codec$ and $1\le r\le m$. Let $$b_r=\sum_{\substack{i=1\\ i\neq r}}^m\left(\alpha_i(\vek{c}_0)-\beta_i(\vek{c}_0)\right).$$ Then $\alpha_r(\vek{c})+b_r=\beta_r(\vek{c})$ for all $\vek{c}\in\codec$ and $1\le r\le m$. It, therefore, follows that $\psi_r=\omega^{b_r}\phi_r$ for all $1\le r\le m$.
\end{proof}
\begin{thm}\label{converse}
    Let $H$ be a normalized BH$(q^k,p)$ matrix, where $p$ is a prime number and $q=p^r$ for some positive integer $r$. Let $\mathscr{C}\subseteq\ff^n$ be a classical linear code over $\ff$ of dimension $k$, where $1\le k<n$, and $\coded=\{(\lambda,\ldots,\lambda):\lambda\in\fff\}\subset\fff^m$ be the one-dimensional linear code over $\fff$. If the quantum code $Q_H(\codec,\coded)$ is a stabilizer code, then $H$ is equivalent to the $rk$-fold Kronecker product of the Fourier matrix of order $p$.
\end{thm}
\begin{proof}
Let $\codec=\{\vek c_1=0,\vek c_2,\ldots,\vek c_{q^k}\}$. We can write $H=[\omega^{f_i(\vek c_j)}]_{1\le i,j\le q^k}$, where $\omega=e^{2\pi i/p}$ and  $f_i:\codec\rightarrow \f$ is a function for each $1\le i\le q^k$. Now, $Q_H(\codec,\coded)$ is the linear span of $\{\phi_1^{\otimes m},\ldots,\phi_{q^k}^{\otimes m}\}$ over $\ff$, where $$\phi_i=\frac{1}{\sqrt{q^k}}\sum_{j=1}^{q^k}\omega^{f_i(\vek c_j)}\ket{\vek c_j}\in(\mathbb{C}^q)^{\otimes n}.$$ Suppose that $Q=Q_H(\codec,\coded)$ is a stabilizer code and let $\s=\stab(Q)$. By Lemma \ref{lem:stab}, $(\vek 0,\ldots,\vek 0\mid \vek d_1,\ldots,\vek d_m)\in\sbar$ for all $\vek d_1,\ldots,\vek d_m\in \codec^\perp$. Let $\vek s=\seq{m}{u}{v}\in\sbar$. Then $\vek s$ is sympletically orthogoanl to the elements of $\sbar$ of the  forms $(\vek 0,\ldots,\vek 0\mid \vek d,\vek 0,\ldots,\vek 0),\ldots,(\vek 0,\ldots,\vek 0\mid \vek 0,\ldots,\vek 0,\vek d)$ for all $\vek d\in\codec^\perp$. Thus, $\tr_{q/p}(\vek u_i.d)=0$ for all $1\le i\le m$ and $\vek d\in\codec^\perp$. Therefore, $\vek u_i\in\codec$ for all $1\le i\le m$ by Lemma \ref{lem:1}. Note that there exists $h\in\f$ such that $\omega^hX(\vek u_1,\ldots,\vek u_m)Z(\vek v_1,\ldots,\vek v_m)\phi_j^{\otimes m}=\phi_j^{\otimes m}$ for all $1\le i\le m$ and $\vek u\in\codec^\perp$. 
This gives that $\omega^h X(\vek u_1)Z(\vek v_1)\phi_j\otimes\cdots\otimes X(\vek u_m)Z(\vek v_m)\phi_j=\phi_j\otimes\cdots\otimes\phi_j$, and so there exists $h_{ij}\in\f$ ($1\le i \le m$, $1\le j\le q^k$) such that $X(\vek u_i)Z(\vek v_i)\phi_j=\omega^{h_{ij}}\phi_j$ for all $1\le i\le m$ and $1\le j\le q^k$ by Lemma \ref{lem:equality_of_tensors}. Since $X(\vek u_i)Z(\vek v_i)\phi_j=\sum_{\vek c\in\codec}\omega^{f_j(\vek c-\vek u_i)}\omega^{\tr_{q/p}(\vek v_i.\vek c)}\ket{\vek c}$, we have 
\begin{equation}\label{reverse-1}
f_j(\vek c-\vek u_i)+\tr_{q/p}(\vek v_i.\vek c)=f_j(\vek c)+h_{ij}
\end{equation}
for all $\vek c\in\codec$, $1\le i\le m$, and $1\le j\le q^k$. Since $H$ is assumed to be normalized, $f_1(\vek c)=0$ for all $\vek c\in\codec$ and $f_i(0)=0$ for all $1\le i \le q^k$. In parrticular, \eqref{reverse-1} yields $\tr_{q/p}(\vek v_i.\vek c)=h_{i1}$ for all $\vek c\in\codec$ and $1\le i\le m$. Substituting $\vek c=0$ and $j=1$ in \eqref{reverse-1}, we obtain $h_{i1}=0$  for all $1\le i\le m$. Hence $\vek v_i\in\codec^\perp$ for all $1\le i\le m$ by Lemma \ref{lem:1}.  Since $\sum_{i=1}^m h_{ij}=-h$, this also shows that $h=0$. Note that  \eqref{reverse-1} turns into 
    $f_j(\vek c-\vek u_i)=f_j(\vek c)+h_{ij}$,
where $h_{ij}=f_j(-\vek u_i)=-f_j(\vek u_i)$ for all $1\le i\le m$ and $1\le j\le q^k$, and $\sum_{i=1}^m f_j(\vek u_i)=0$ for all $1\le j\le q^k$. It follows that 
\begin{equation}\label{reverse-3}
    f_j(\vek c-\vek u_i)=f_j(\vek c)-f_j(\vek u_i)
\end{equation}
for all $\vek c\in\codec$, $1\le i\le m$, and $1\le j\le q^k$. Replacing $\vek c$ by $\vek{c}+\vek{u}_i$ in \eqref{reverse-3}, we obtain that $f_j(\vek c+\vek u_i)=f_j(\vek c)+f_j(\vek u_i)$ for all $\vek c\in\codec$, $1\le i\le m$, and $1\le j\le q^k$. In particular, we have $f_j(\sum_{i=1}^m \vek u_i)=\sum_{i=1}^m f_j(\vek v_i)=0$ for all $1\le j\le q^k$. Since $H$, whose rank is $q^k$, has its first column consisting of $1$'s, this is possible only when $\sum_{i=1}^m \vek u_i=0$. It follows that $\sbar$ is contained in $$\mathcal{A}:=\{\seq{m}{u}{v}:\vek v_i\in\codec^\perp,\ \vek u_i\in\codec,\ \forall 1\le i\le m \text{ with } \sum_{i=1}^m \vek u_i=0\}.$$ Rewriting $\mathcal{A}$ as $$\{(\vek u_1,\ldots\vek u_{m-1},-\sum_{i=1}^{m-1} \vek u_i\mid \vek v_1,\ldots,\vek v_m):\vek u_1,\ldots\vek u_m\in\codec,\ \vek v_1,\ldots,\vek v_m\in\codec^\perp\},$$ we have $|\sbar|=q^{nm-k}=(q^k)^{m-1}(q^{n-k})^m=|\mathcal{A}|$; hence $\sbar=\mathcal{A}$. Now \eqref{reverse-3} gives that $f_j(\vek c-\vek c')=f_j(\vek c)-f_j(\vek c')$ for all $\vek c,\vek c'\in\codec$ and $1\le j\le q^k$, proving that $f_j:\codec\rightarrow\f$ is a linear transformation of $\f$-spaces.
\end{proof}

\appendix
\section{Appendix: Equivalence  of Fourier-type BH Matrices}
Throughout this brief note, $R$ will denote a finite Frobenius ring
and $M$ a finite $R$-bimodule. We also denote the set of all non-degenerate
bilinear forms on $M$ by $\blf(M)$. 

Let $B:\bil M$ be a non-degenerate bilinear form on $M$. Associated
to $B$ are there right $R$-module homomorphisms, for all $x\in M$,
defined by
\[
\xymatrix{B(x): & M\ar[r] & R\\
 & y\ar@{|->}[r] & B(x,y)
}
\]
that is $B(x)(y)=B(x,y)$ for all $x,y\in M$. It follows that $B(x)\in M^{*}=\Hom(M_{R},R_{R})$
for every $x\in M$. Note the following properties:
\begin{itemize}
\item $B(x_{1}+x_{2})=B(x_{1})+B(x_{2})$ for all $x_{1},x_{2}\in M$.
\item $B(rx)=rB(x)$ in the left $R$-module $M^{*}$for all $r\in R$ and
$x\in M$.
\item $B(x)=B(x')$ if and only if $x=x'$ for all $x,x'\in M$.
\item $M^{*}=\{B(x):x\in M\}$.
\end{itemize}
Let $\Aut(_{R}M)$ (respectively, $\Aut(M_{R})$) denote the group
of left (respectively, right) $R$-module automorphisms of $M$, equipped
with the usual composition of maps. Given $\gamma\in\Aut(_{R}M)$
and $\eta\in\Aut(M_{R})$, one can define two mappings assosiated
with $B$ as follows:
\[
\xymatrix{B':M\times M\ar[rr] &  & R\\
(x,y)\ar@{|->}[rr] &  & B(\gamma(x),y)
}
\]
\[
\xymatrix{B'':M\times M\ar[rr] &  & R\\
(x,y)\ar@{|->}[rr] &  & B(x,\eta(y))
}
\]
Observe that both $B'$ and $B''$ are non-degenerate bilinear forms.
Thus both groups $\Aut(_{R}M)$ and $\Aut(M_{R})$ act on $\blf(M)$.
We write $B'=B\cdot\gamma$ and $B''=B\cdot\eta$.
\begin{thm}\label{theo A.1}
Let $B$ be any non-degenerate bilinear form on $M$. Then we have
\begin{align*}
\blf(M) & =\{B\cdot\gamma:\gamma\in\Aut(_{R}M)\}\\
 & =\{B\cdot\eta:\eta\in\Aut(M_{R})\}.
\end{align*}
\end{thm}

\begin{proof}
The arguments above the theorem show that $\{B\cdot\gamma:\gamma\in\Aut(_{R}M)\}$
lies in $\blf(M)$. 

Conversely, let $B'$ be any other non-degenerate bilinear form on
$M$. We shall show that $B'=B\cdot\gamma$ for a suitable $\gamma\in\Aut(_{R}M)$. 

Let $M=\{0,x_{1},\ldots,x_{n}\}$. Since 
\begin{equation*}
    M^{*}=\{B(0),B(x_{1}),\ldots,B(x_{n})\}=\{B'(0),B'(x_{1}),\ldots,B'(x_{n})\},
\end{equation*}
there exists $\sigma\in S_{n}$ such that $B'(x_{i})=B(x_{\sigma(i)})$.
Now define $\gamma:M\rightarrow M$ by $\gamma(0)=0$ and $\gamma(x_{i})=x_{\sigma(i)}$
for each $i=1,\ldots,n$. Notice that 
\[
B'(x,y)=B(\gamma(x),y)
\]
for all $x,y\in M$. Thus we complete the proof by showing that $\gamma\in\Aut(_{R}M)$.

Let $x_{i}+x_{j}=x_{k}$. Since 
\begin{align*}
B(x_{\sigma(k)}) & =B'(x_{k})\\
 & =B'(x_{i}+x_{j})\\
 & =B'(x_{i})+B'(x_{j})\\
 & =B(x_{\sigma(i)})+B(x_{\sigma(j)})\\
 & =B(x_{\sigma(i)}+x_{\sigma(j)})
\end{align*}
we have $x_{\sigma(i)}+x_{\sigma(j)}=x_{\sigma(k)}$. It follows that
\[
\gamma(x_{i}+x_{j})=\gamma(x_{k})=x_{\sigma(k)}=x_{\sigma(i)}+x_{\sigma(j)}=\gamma(x_{i})+\gamma(x_{j}),
\]
i.e. $\gamma$ is additive. On the other hand if $r\in R$ and $rx_{i}=x_{j}$,
then
\begin{align*}
B(x_{\sigma(j)}) & =B'(x_{j})\\
 & =B'(rx_{i})\\
 & =rB'(x_{i})\\
 & =rB(x_{\sigma(i)})\\
 & =B(rx_{\sigma(i)}),
\end{align*}
hence $x_{\sigma(j)}=rx_{\sigma(i)}$, which yields
\[
\gamma(rx_{i})=\gamma(x_{j})=x_{\sigma(j)}=rx_{\sigma(i)}=r\gamma(x_{i}),
\]
as desired. This completes the proof of the first equality. By symmetric
arguments, one easily prove the other equality. 
\end{proof}
\begin{prop}\label{prop A.2}
Let $\chi$ be a generating character of $R$ and let $B$, $B'$
be two non-degenerate bilinear forms of the $R$-bimodule $M=\{x_{0}=0,x_{1},\ldots,x_{n}\}$.
Then the matrices 
\[
H=[\chi(B(x_{i},x_{j}))]_{0\le i,j\le n}
\]
and
\[
H'=[\chi(B'(x_{i},x_{j}))]_{0\le i,j\le n}
\]
are equivalent (by row permutation).
\end{prop}

\begin{proof}
There exists an $\sigma\in S_n$ such that $B'(x_i,x_j)=B(x_{\sigma(i),x_j})$ for all $0\leq i,j\leq n$ by the proof of Theorem \ref{theo A.1}. Thus $H'=[\chi(B'(x_{i},x_{j}))]=[B(x_{\sigma(i)},x_j)] 
$ is $H$ with rows permuted by $\sigma$.
\end{proof}
\begin{lem}
Let $\chi$ and $\chi'$ be two generating characters of the ring
$R$. Then there exists a unit element $a\in R$ such that
$\chi'=\chi^{a}$.
\end{lem}

\begin{proof}
Since $\chi$ is a generating character then there exists an element $r$ of  $R$ such that $\chi'=\chi r$. Similarly, there exists an element $r'$ of $R$ such that $\chi=\chi'r'$. Therefore $\chi(1-rr')=0$ and $\chi\neq 0$. Then $\chi'=\chi r=\chi^r$  and $r$ is a unit element of $R$.  

\end{proof}
\begin{prop}\label{prop A.4}
Let $\chi$ and $\chi'$ be two generating characters of the ring
$R$ and let $B$ be a non-degenerate bilinear form on the $R$-bimodule
$M=\{x_{0}=0,x_{1},\ldots,x_{n}\}$. Then the matrices 
\[
H=[\chi(B(x_{i},x_{j}))]_{0\le i,j\le n}
\]
and
\[
H'=[\chi'(B(x_{i},x_{j}))]_{0\le i,j\le n}
\]
are equivalent (by row permutation).
\end{prop}

\begin{proof}
By above lemma, there exists a unit element $a\in R$ such that $\chi'=\chi^{a}$.
Since $a$ is unit, we have $aM=M$. It follows that there exists
$\sigma\in S_{n}$ such that $ax_{i}=x_{\sigma(i)}$ for all $i=1,\ldots,n$.
Now
\begin{align*}
    H'&=[\chi'(B(x_{i},x_{j}))]=[\chi^{a}(B(x_{i},x_{j}))]\\
    &=[\chi(aB(x_{i},x_{j}))]=[\chi(B(ax_{i},x_{j}))]\\
    &=[\chi(B(x_{\sigma(i)},x_{j}))]
\end{align*}
is clearly the matrix $[\chi(B(x_{i},x_{j}))]$ with rows permuted
by $\sigma$.
\end{proof}
\begin{prop}\label{prop:eq}
For the $R$-bimodule $M=\{x_{0}=0,x_{1},\ldots,x_{n}\}$, the matrices
of the form
\[
[\chi(B(x_{i},x_{j}))]_{0\le i,j\le n},
\]
where $\chi$ is a generating character of $R$ and $B:M\times M\rightarrow R$
is a non-degenerate bilinear form on $M$, are all equivalent.
\end{prop}

\begin{proof}
For every non-degenerate bilinear form $B'$ on $M$ distinct from $B$ according to Proposition \ref{prop A.2}, $H'=[\chi(B'(x_{i},x_{j}))]_{0\le i,j\le n}$ is equivalent to $H$. Similarly for every $\chi'$ distinct from $\chi$ according to Proposition \ref{prop A.4}, $H''=[\chi'(B(x_{i},x_{j}))]_{0\le i,j\le n}$ is equivalent to $H$. This completes the proof. 
\end{proof}
\begin{cor}\label{last corollary}
   Let $H=[\omega^{f_i(\vek c_j)}]_{1\le i,j\le q^k}$, where $\omega=e^{2\pi i/p}$ and  $f_i:\codec\rightarrow \f$ is a linear transformation for each $1\le i\le q^k$. Then $H$ is equivalent to a Kronecker product of Fourier matrices.
\end{cor}
\begin{proof}
 Since $f_i\in\mathcal{L}(\codec,\f)$ for $1\leq i,j\leq q^k$ and $\codec$ is isomorphic to $\fff$ then we use the composition of $\f$-space isomorphisms $\kappa:\codec \rightarrow \fff\rightarrow\lt$.  Let $B:\codec\times\codec\rightarrow\f$ be a transformation defined as $B(c,c')=f_i(c')$ such that $\kappa(c)=f_i$. Then $B$ is a non-degenerate bilinear form. Also $\chi: \f\rightarrow \mathbb{C}^*, \chi(a)=\omega^a$ is a generating character for $\f$.  If we combine these we say 
 \begin{equation*}
     H=[\omega^{f_i(\vek c_j)}]_{1\le i,j\le q^k}=[\chi(B(c_i,c_j))]_{1\leq i,j\leq q^k}.
 \end{equation*}
  On the other hand $B':\fff\times\fff\rightarrow\fff, B(i,j)=ij$ is a non-degenerate bilinear form and $\chi':\fff\rightarrow\mathcal{C}^*, \chi'(a)=\omega^a$ is a generating character. Then the Fourier matrix $F_{ij}=[\omega^{ij}]=[\chi'(B'(i,j))]_{1\leq i,j\leq q^k}$. The desired result is obtained by Proposition $\ref{prop:eq}$. 
\end{proof}

\bibliographystyle{plain}
\bibliography{References}
\end{document}